\documentclass[journal,12pt,onecolumn,draftclsnofoot]{IEEEtran}

\usepackage{aligned-overset}
\usepackage[dvipsnames]{xcolor}
\usepackage{amsthm}

\usepackage{amsmath}%
\usepackage{MnSymbol}%
\usepackage{wasysym}
\usepackage{authblk}
\usepackage{subcaption}
\usepackage[ruled,vlined]{algorithm2e}

\usepackage[mathscr]{euscript}
\usepackage{cite}
\ifCLASSINFOpdf
\usepackage[pdftex]{graphicx}
\graphicspath{{../pdf/}{../jpeg/}}
\DeclareGraphicsExtensions{.pdf,.jpeg,.png}

\else
\usepackage[dvips]{graphicx}
\graphicspath{{../eps/}}
\DeclareGraphicsExtensions{.eps}
\fi
\usepackage{amsmath}
\usepackage{array}

\usepackage{tikz}
\usepackage[utf8]{inputenc}
\usepackage{pgfplots} 
\usepackage{pgfgantt}
\usepackage{pdflscape}
\pgfplotsset{compat=newest} 
\pgfplotsset{plot coordinates/math parser=false}
\usepackage{grffile}
\usetikzlibrary{plotmarks}
\usepackage{amsthm}
\usepackage{amsmath}%
\usepackage{MnSymbol}%
\usepackage{wasysym}
\usepackage{subcaption}

\usepackage[mathscr]{euscript}
\usepackage{cite}
\usepackage{amsfonts}
\usepackage{pgf}
\usepackage{mathtools}
\usepackage{graphicx}
\usepackage{array}
\usepackage{tikz}
\usepackage[utf8]{inputenc}
\usepackage{pgfplots} 
\usepackage{pgfgantt}
\usepackage{pdflscape}
\pgfplotsset{compat=newest} 
\pgfplotsset{plot coordinates/math parser=false}
\usepackage{grffile}
\usetikzlibrary{plotmarks}
\usepackage{tikzscale}
\usepackage{amsmath}
\usepackage{array}
\usepackage{tikz}
\usepackage[utf8]{inputenc}
\usepackage{pgfplots} 
\usepackage{pgfgantt}
\usepackage{pdflscape}
\pgfplotsset{compat=newest} 
\pgfplotsset{plot coordinates/math parser=false}
\usepackage{grffile}
\usetikzlibrary{plotmarks}
\usepackage[T1]{fontenc}
\usepackage{graphicx}
\usepackage{overpic}
\usepackage{array}
\usepackage{color}
\usepackage{url}
\usepackage{tikz}
\usepackage[utf8]{inputenc}
\usepackage{pgfplots} 
\usepackage{cite}

\newtheorem{theorem}{Theorem}

\def\CN{\mathcal{C}\mathcal{N}} %Complex Gaussian
\usepackage[colorlinks=true,bookmarks=false,citecolor=blue,urlcolor=blue]{hyperref}  % hyperref still needs to be put at the end!

\IEEEoverridecommandlockouts
\begin{document}
	\title{Energy Efficiency Maximization in the Uplink Delta-OMA Networks}
% 	\vspace{-3mm}
	\author{Ramin~Hashemi,~\IEEEmembership{Student Member,~IEEE,}
		Hamzeh Beyranvand,~\IEEEmembership{Member,~IEEE}, Mohammad Robat Mili, Ata Khalili, \IEEEmembership{Member,~IEEE}, 
		Hina Tabassum, \IEEEmembership{Senior Member,~IEEE}, and Derrick Wing Kwan Ng, \IEEEmembership{Fellow, IEEE}

		\vspace{-1mm}
		\thanks{
		
		R. Hashemi and H. Beyranvand are with the  Department of Electrical Engineering, Amirkabir University of Technology, 424 Hafez Avenue, Tehran 15914, Iran, e-mails: (raminhashemi@aut.ac.ir, beyranvand@aut.ac.ir).
		M R. Mili is with the Department of Telecommunications and information processing, Ghent University, Belgium, e-mail: (mohammad.robatmili@ieee.org).
		A. Khalili is with the Electronics Research Institute, Sharif University of Technology, Tehran, Iran (ata.khalili@ieee.org). 
		H. Tabassum is with the Lassonde School of Engineering at York University, Canada (e-mail:hina@eecs.yorku.ca).	
        D. W. K. Ng is with the School of Electrical Engineering and Telecommunications, University of New South Wales, Sydney, NSW 2052, Australia (e-mail: w.k.ng@unsw.edu.au).

	}}
	
	\maketitle
	\vspace{-3mm}
	\begin{abstract}
        Delta-orthogonal multiple access (D-OMA) has been recently investigated as a potential technique to enhance the spectral efficiency in the sixth-generation (6G) networks. D-OMA enables partial overlapping of the adjacent sub-channels that are assigned to different clusters of users served by non-orthogonal multiple access (NOMA), at the expense of additional interference. In this paper, we analyze the performance of D-OMA in the uplink and develop a multi-objective optimization framework to maximize the uplink energy efficiency (EE) in a \textcolor{black}{multi-access point (AP)} network enabled by D-OMA. Specifically, we optimize the  sub-channel and transmit power allocations of the users as well as the overlapping percentage of the spectrum between the adjacent sub-channels. The formulated problem is a mixed binary non-linear programming problem. Therefore, to address the challenge we first transform the problem into a single-objective problem using Tchebycheff method. Then, we apply the monotonic optimization (MO) to explore the hidden monotonicity of the objective function and constraints, and reformulate the problem into a standard MO in canonical form. The reformulated problem is then solved by applying the outer polyblock approximation method. Our numerical results show that D-OMA outperforms the conventional non-orthogonal multiple access (NOMA) and orthogonal frequency division multiple access (OFDMA)  when the adjacent sub-channel overlap and scheduling are optimized jointly.
	\end{abstract}
	\begin{IEEEkeywords}
		Delta-OMA, multi-objective optimization, power control, resource allocation.  
	\end{IEEEkeywords}
	\vspace{-6mm}
	\IEEEpeerreviewmaketitle
	
	% 	\IEEEPARstart{T}{he}
	\section{Introduction}
	\bstctlcite{IEEEexample:BSTcontrol}
Delta orthogonal multiple access (D-OMA) has been recently considered as a potential variant of non-orthogonal multiple access (NOMA) to enable massive multiple access and enhanced spectral efficiency in beyond-5G and6G networks~\cite{Al-Eryani2019,David2018,al2019delta}. D-OMA exploits partial overlapping of adjacent sub-channels that are assigned to different clusters of users served by NOMA and thereby enhance spectral efficiency. That is, NOMA is a special case of D-OMA when there is no overlapping of adjacent sub-channels. Clearly, the performance of D-OMA critically depends on the number of users in  a NOMA cluster, the fraction of overlapping spectrum, and sub-channel scheduling. It is noteworthy that while partial overlapping of adjacent sub-channels may enhance spectral efficiency, it can yield  additional interference that can result in significant performance loss.
Therefore, it is thus crucial to optimize the scheduling, NOMA cluster size, and the fraction of overlapping spectrum  efficiently.

To date, many research works have considered optimizing the performance of stand-alone NOMA or hybrid NOMA-OMA networks \cite{ali2016dynamic,Zeng2019,Shi2019,Sun2017,Wei2020}, and partial NOMA (P-NOMA) \cite{Kim2019,Ali2020}. The fundamental insights on the gains of NOMA over OMA are elaborated in \cite{Wei2020}. 
In hybrid NOMA-OMA, the transmit power, time, and sub-channel resources to the users are determined to efficiently exploit both NOMA and OMA. Compared to the traditional NOMA or OMA, hybrid NOMA has a variety of benefits, including higher spectral efficiency than OMA, less complexity in terms of successive interference cancellation (SIC) than NOMA, and reduced interference than NOMA.

In \cite{ali2016dynamic}, a user grouping and power allocation strategy based on sum-rate maximization was proposed for both the uplink (UL) and downlink NOMA. 
% The authors in \cite{Chen2014} noted that employing hybrid NOMA systems would be applicable to 5G communications where the sum rate of the network along with sub-band allocation is taken into account. 
Also, in \cite{Zeng2019}, an energy-efficient power and resource block allocation framework is presented for the UL  of a hybrid NOMA network with the quality of service (QoS) constraints. In \cite{Shi2019}, the downlink energy efficiency (EE)  of the network is maximized by optimizing a user clustering and power control framework in a hybrid NOMA system. \textcolor{black}{On the other hand, P-NOMA partially overlaps the signals of the users by adjusting the extent of the overlap, thereby reducing the interference from other users. In particular, the P-NOMA setup is for two users by introducing two parameters which are positive real-valued numbers varying between zero and one thus, the spectrum is divided between overlap and the amount of non-overlap regions to each user \cite{Kim2019,Ali2020}. The motivation behind P-NOMA is to introduce higher flexibility into the system by having control over how much of the spectrum can overlap for two users.}

Different from the aforementioned variants of NOMA, in D-OMA, the spectrum overlapping is considered among two NOMA clusters operating on adjacent sub-channels within a given access point and the interference is controlled by optimizing either the fraction of overlapping percentage or reducing the cluster size.  Very recently, the authors in \cite{Al-Eryani2019} have shown preliminary results on the significance of D-OMA in the downlink, compared to  NOMA \cite{Tabassum2016}.

In this paper, we provide a comprehensive framework to analyze the performance  of D-OMA in the UL. Specifically, we develop a multi-objective optimization framework to maximize the uplink EE  of a multi-access point network enabled by D-OMA.~\textcolor{black}{In general, EE maximization problems are more important in the UL of wireless networks where extending the battery life of mobile users is of great concern.} We optimize the  sub-channel and transmit power allocations of the users as well as the overlapping percentage of the spectrum between the adjacent sub-channels. The formulated problem is a mixed binary non-linear programming problem, therefore, we first transform the problem into a single-objective problem using Tchebycheff method. Then, we apply the monotonic optimization (MO) framework to explore the hidden monotonicity of the objective function and constraints in order to reformulate the problem into a standard MO in canonical form. The re-formulated problem is then solved by the outer polyblock approximation method. The numerical results show that D-OMA method outperforms the conventional NOMA and OMA when the adjacent sub-channel overlap and scheduling is optimized jointly. {In addition, our numerical results depict  the effectiveness of D-OMA considering two cases: \textbf{(i)} when the overlapping percentages are optimized individually on each sub-channel denoted as per-sub-channel optimized delta (POD), and \textbf{(ii)}~when the overlapping parameter is fixed for all sub-channels denoted as non-POD (NPOD).}

    \vspace{-4mm}
	\section{System Model and Problem Formulation}
	\subsection{System Model}
	\textcolor{black}{Consider the UL of a multi-access point (AP) network shown in Fig. \ref{fig:SystemModel_Network} where the D-OMA method is leveraged to serve the users within $K$ APs.}
	\begin{figure*}[t]
		\centering
		\begin{subfigure}[b]{0.56\textwidth}  
			\centering
			\includegraphics[scale = 0.5]{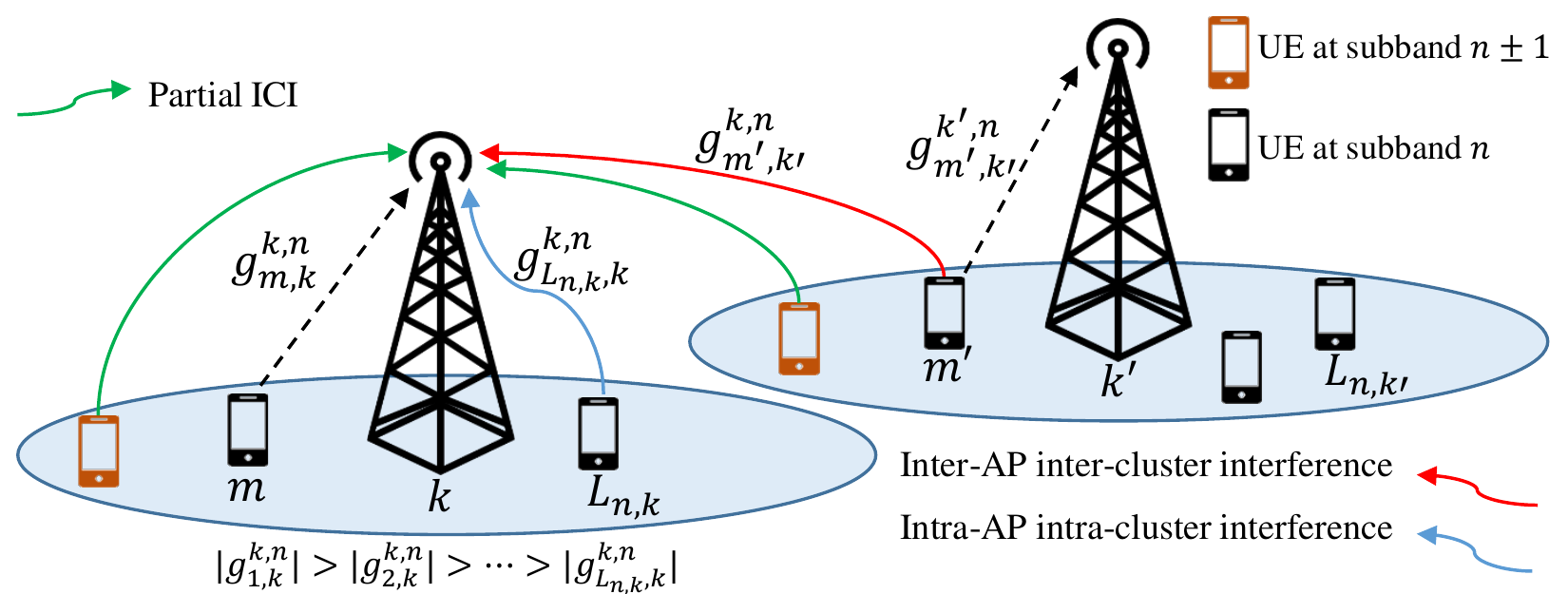}
			\caption{\textcolor{black}{Network topology and illustration of interference signals.}}
			\label{fig:SystemModel_Network}
		\end{subfigure}
		\begin{subfigure}[b]{0.43\textwidth}
			\centering    
			\includegraphics[width=2.8in]{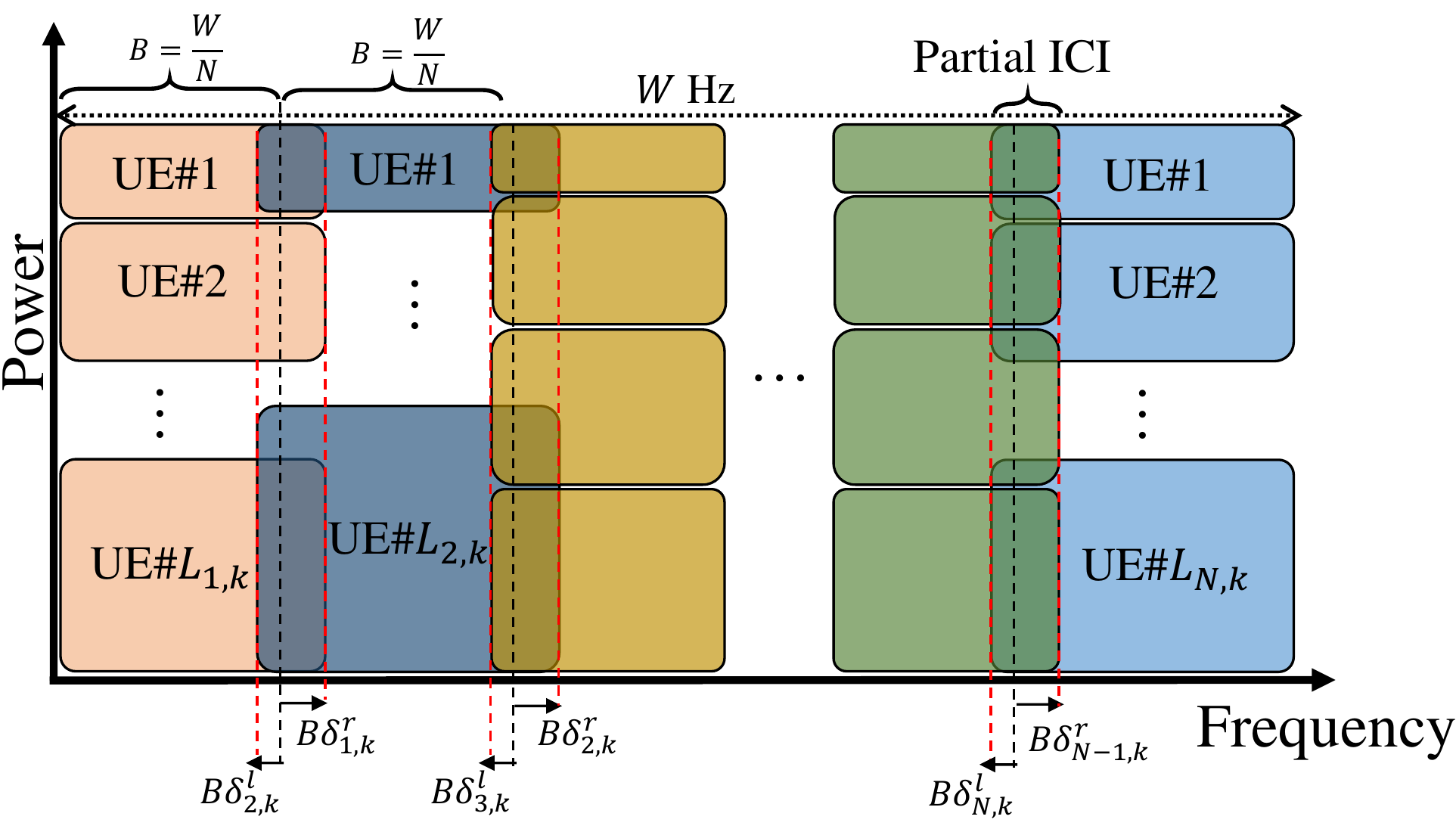}
			\caption{D-OMA and partial ICI illustration in AP $k$.}
			\label{fig:SystemModel}
		\end{subfigure}
		\caption{The UL D-OMA transmission.}
	\end{figure*}
	The principal methodology of D-OMA scheme is illustrated in Fig. \ref{fig:SystemModel} in which the overall system bandwidth ($W$ Hz) is divided into $N$ subbands and each subband is allocated to a sub-set of users (i.e., users in a specific NOMA cluster)\footnote{\textcolor{black}{D-OMA allows different NOMA clusters with  adjacent frequency bands to overlap by an amount of $\delta$ percent of their maximum allocated sub-channel. Consequently, the spectral efficiency achieved by massive in-band NOMA cluster is maintained by adding more clusters within the same allocated overall spectrum.}}.  Assume the number of users in subband $n$ at AP $k$ as $L_{n,k}$, and $\mathscr{U}_k=\{1,2,...,U_k\}$ \textcolor{black}{denotes} the set of users in the coverage \textcolor{black}{area} of AP $k$. Thus, we have $\sum_{n=1}^{N}L_{n,k}=U_k$. We indicate the set of subbands by $\mathscr{N}$ and the set of all APs by $\mathscr{K}$.

	In D-OMA method, an additional interference signal named as partial inter-cluster interference (ICI) from the adjacent sub-channels incur. \textcolor{black}{This interference will be observed from users in  the same AP as well as other APs.} {The adjacent subbands are interleaved from the right- and left-side by amount of $\text{B}\times\delta_{n,k}^r$ Hz and $\text{B}\times\delta_{n,k}^l$ Hz, respectively, where $0\leq \delta_{n,k}^r\leq 1$ and $0\leq \delta_{n,k}^l\leq 1$ denote the overlap percentage of subband $n$ for AP $k$, and $\text{B}$ is the bandwidth of each subband i.e. $\text{B}=\frac{W}{N}$ (see Fig. \ref{fig:SystemModel}). Note that when $\delta_{n,k}^r=  \delta_{n,k}^l=1$ the subbands are completely overlapped such that the amount of interference power is maximum. On the other hand, each subband's effective bandwidth denoted by $\text{B}_n=\text{B}(1+\delta_{n,k}^r+\delta_{n,k}^l)$ is expanded three times at the expense of additional  ICI. } Subsequently, the optimal values of $\delta_{n,k}^r$ and $\delta_{n,k}^l$ for $\forall n \in \mathscr{N}$ and $\forall k \in \mathscr{K}$ and cluster size should be determined efficiently to enhance the network sum rate (SR). Note that \textcolor{black}{$\delta_{n,k}^r=0$ and $\delta_{n,k}^l=0$} for $\forall n \in \mathscr{N}$ and $\forall k \in \mathscr{K}$ correspond to the conventional NOMA-OFDM.

	Let $p_{m,k}$ denote the transmitted data power from the \textcolor{black}{user equipment (UE)} $m$ to the AP $k$. To model the subband allocation, we define the following binary variable
	\begin{flalign}
		\rho_{m,k}^n = 
		\begin{cases}
			1, &\smallskip\text{if UE } m \text{ in AP } k \text{ is associated with subband } n, \\    
			0. &\smallskip \text{otherwise.}
		\end{cases}
		\nonumber 
	\end{flalign}
	Therefore, the received signal at AP $k$ in subband $n$ will be as given in 
	\begin{flalign}
		y_{k,m}^n = & \underset{\text{Desired signal for UE } m}{\underbrace{\rho_{m,k}^n s_{m,k}\displaystyle\sqrt{p_{m,k}}g_{m,k}^{k,n}}} + \underset{\text{\textcolor{black}{Intra-AP  ICI}}} {\underbrace{\sum_{\substack{m'=1,\\m'\neq m}}^{U_k}\rho_{m',k}^{n} s_{m',k}\displaystyle\sqrt{p_{m',k}}g_{m',k}^{k,n}}}+
		\underset{\text{\textcolor{black}{Inter-AP ICI}}} {\underbrace{\sum_{\substack{k'=1,\\k'\neq k}}^{K}\sum_{\substack{m'=1}}^{U_{k'}}\rho_{m',k'}^{n}s_{m',k'}\displaystyle\sqrt{p_{m',k'}}g_{m',k'}^{k,n}}}  \label{received_signal} \\ 
		& +
		\underset{\text{\textcolor{black}{Partial ICI}}} {\underbrace{\sum_{\substack{k'=1}}^{K}\sum_{\substack{m'=1}}^{U_{k'}}\Big((\sqrt{\delta_{n+1,k'}^{l}}+\sqrt{\delta_{n,k'}^{r}})g_{m',k'}^{k,n+1}\rho_{m',k'}^{n+ 1}+(\sqrt{\delta_{n,k'}^{l}}+\sqrt{\delta_{n-1,k'}^{r}})g_{m',k'}^{k,n-1}\rho_{m',k'}^{n- 1}\Big)s_{m',k'}\displaystyle\sqrt{p_{m',k'}}}} \nonumber \\ & + \sqrt{1+\delta_{n,k}^{l}+\delta_{n,k}^{r}}w^n_k, \nonumber
	\end{flalign}
	\textcolor{black}{ where $w^n_k\sim\CN(0,\sigma^2)$ is the additive complex Gaussian white noise for a given bandwidth B, in which the noise power is increased due to subband expansion to $(1+\delta_{n,k}^{l}+\delta_{n,k}^{r})\sigma^2$}, $s_{m,k}$ represents the transmitted symbol from UE $m$ to the AP $k$ with $\mathbb{E}[|s_{m,k}|^2]=1$, $\delta_{n,k}^r$ and $\delta_{n,k}^l$ are the introduced parameters to control the partial ICI at the subband $n$ in AP $k$. \textcolor{black}{It is worth noting that the ICI  is caused by using the same subband by other nearby clusters. The users within a certain cluster will suffer from ICI resulting from use of the same subband by other nearby clusters. The size of a NOMA cluster is considered a design parameter to reach trade-offs between different factors, namely, the data-rate necessities of the users; the total power budget per NOMA cluster; the complexity level at the NOMA receivers; and the NOMA user immunity to ICI-based and SIC-based error propagation.  On the other hand, SIC allows successive decoding of users' signals.}
	
	\vspace{-3mm}
	\subsection{Network Rate and Energy Efficiency}
	\textcolor{black}{It is important to note that there are four types of partial ICI on a given subband $n$ in AP $k$, as illustrated in Fig. \ref{fig:SystemModel}. For instance, the partial ICI on the left-hand side is due to the expansion of  subband $n-1$ towards the right-hand side which is controlled by $\delta^{r}_{n-1,k}$ as well as the expansion of subband $n$ to the left-hand side where it is controlled by $\delta^{l}_{n,k}$. The same inference applies to the other two partial ICI signals on the right-hand side of subband $n$.} 
	
	Note that, $\delta^{r}_{N,k}$ and $\delta^{l}_{1,k}$ do not exist and therefore their value is zero. $g_{m,k'}^{k,n}$ indicates the channel \textcolor{black}{power gain} between UE $m$ in AP $k'$ and AP $k$ at subband $n$ which is denoted as $g_{m,k'}^{k,n}=h_{m,k'}^{k,n}\sqrt{\beta_{m,k'}^{k,n}}$, where $h_{m,k'}^{k,n}$ is the small-scale fading coefficient which follows circularly-symmetric zero-mean complex normal distribution with unit variance as $h_{m,k'}^{k,n}\sim\CN(0,1)$ and $\beta_{m,k'}^{k,n}$ represents the large-scale fading and path loss. The channel gains are being sorted at APs, i.e., $|g_{1,k}^{k,n}|\geq |g_{2,k}^{k,n}|\geq ...\geq |g_{L_{n,k},k}^{k,n}|$, $\forall k \in \mathscr{K}$, $\forall n \in \mathscr{N}$, to perform SIC, in order to extract the desired signal of all UEs. Therefore, the achievable rate of UE $m$ associated with AP $k$ at subband $n$ is given by
	\begin{flalign}
		R_{m,k}^{n} = \text{B}_n\log_2\bigg(1+\frac{\displaystyle p_{m,k}|g_{m,k}^{k,n}|^2}{ \text{I}_\text{\textcolor{black}{IntraICI}}^{m,n,k} + \text{I}_\text{\textcolor{black}{InterICI}}^{m,n,k}+ \text{I}_\text{\textcolor{black}{PartialICI}}^{m,n,k}+\tilde{\sigma}^2_{n,k}}\bigg),
		\label{achievable_rate_conventional_}
	\end{flalign}
	where $\text{B}=\frac{W}{N}$, $\text{B}_n=\text{B}(1+\delta_{n,k}^r+\delta_{n,k}^l)$,  $\tilde{\sigma}^2_{n,k}=\sigma^2(1+{\delta_{n,k}^{l}}+{\delta_{n,k}^{r}})$, and $\text{I}_\text{\textcolor{black}{IntraICI}}^{m,n,k}$, $\text{I}_\text{\textcolor{black}{InterICI}}^{m,n,k}$ and $\text{I}_\text{\textcolor{black}{PartialICI}}^{m,n,k}$ are defined as
	\begin{flalign}
		\text{I}_\text{\textcolor{black}{IntraICI}}^{m,n,k}=& \displaystyle\sum_{m'=m+1}^{U_k}\rho_{m',k}^n p_{m',k}|g_{m',k}^{k,n}|^2,
		 \nonumber\\
		\text{I}_\text{\textcolor{black}{InterICI}}^{m,n,k}= & \sum_{\substack{k'=1, k'\neq k}}^{K}\sum_{\substack{m'=1}}^{U_{k'}}\rho_{m',k'}^{n}p_{m',k'}|g^{k,n}_{m',k'}|^2, \label{IntraInt} \\ 
		\text{I}_\text{\textcolor{black}{PartialICI}}^{m,n,k} = & \sum_{\substack{k'=1}}^{K}\sum_{\substack{m'=1}}^{U_{k'}}\Big((\sqrt{\delta_{n+1,k'}^{l}}+\sqrt{\delta_{n,k'}^{r}})^2|g_{m',k'}^{k,n+1}|^2\rho_{m',k'}^{n+ 1}  + (\sqrt{\delta_{n,k'}^{l}}+\sqrt{\delta_{n-1,k'}^{r}})^2|g_{m',k'}^{k,n-1}|^2\rho_{m',k'}^{n- 1}\Big)\displaystyle{p_{m',k'}}. \nonumber
% 		\label{I_ICI}
	\end{flalign}
	\textcolor{black}{where $\text{I}_\text{\textcolor{black}{IntraICI}}^{m,n,k}$ and $\text{I}_\text{\textcolor{black}{InterICI}}^{m,n,k}$ are the power of the \textcolor{black}{\textit{intra-AP  intra-cluster} interference and \textit{inter-AP inter-cluster} interference terms} in \eqref{received_signal}, respectively. Additionally, $\text{I}_\text{\textcolor{black}{PartialICI}}^{m,n,k}$ refer to the power of partial ICI components as a function of $\delta_{n,k}^r$ and $\delta_{n,k}^l$.} For simplicity, let us define $\text{I}_\text{\textcolor{black}{Total}}^{m,n,k}=\text{I}_\text{\textcolor{black}{IntraICI}}^{m,n,k} + \text{I}_\text{\textcolor{black}{InterICI}}^{m,n,k}+ \text{I}_\text{\textcolor{black}{PartialICI}}^{m,n,k}+\tilde{\sigma}^2_{n,k}$ as the total interference plus noise power. The total SR of network in bps/Hz is given by $\text{SR} = \displaystyle \sum\limits_{n=1}^{N}\sum\limits_{k=1}^{K}\sum\limits_{m=1}^{U_k} \rho_{m,k}^n R_{m,k}^n$	. It is inferred that by appropriately choosing \textcolor{black}{$\delta_{n,k}^r$ and $\delta_{n,k}^l$} for $\forall n \in \mathscr{N}, \forall k \in \mathscr{K}$, the SR will be increased as well. To the best of our knowledge, the joint optimization of SR and sum power (SP) in terms of finding optimal value of $\boldsymbol{\delta}$, $\textbf{p}$, and $\boldsymbol{\rho}$ have not been investigated before where we will discuss it in  \textcolor{black}{the} next subsequent sections. Note that $\boldsymbol{\rho}$, $\textbf{p}$ and $\boldsymbol{\delta}$ \textcolor{black}{are the vector representation of} the variables $p_{k,m}$ for $\forall (m,k)\in\mathscr{U}_k\times\mathscr{K}$ and $\rho_{m,k}^n$, $\forall (m,k,n)\in\mathscr{U}_k\times\mathscr{K}\times\mathscr{N}$ and \textcolor{black}{$\delta_{n,k}^r$ and $\delta_{n,k}^l$ for} $\forall (n,k)\in\mathscr{N}$\textbackslash$N \times \mathscr{K}$, respectively. 
	
	Our aim is to optimize the total energy efficiency (EE) which is  $\text{EE} = \frac{\text{SR}}{\text{SP}+ \text{CP}}$ where $\text{CP}=\sum\limits_{k=1}^{K}\sum\limits_{m=1}^{U_k}p^{\text{circuit}}_{m,k}$ denotes the total circuit power consumption with $p^{\text{circuit}}_{m,k}$ for UE $m$ in AP $k$ and $\text{SP}=\sum\limits_{k=1}^{K}\sum\limits_{m=1}^{U_k}p_{m,k}$ is the total transmitted data power \textcolor{black}{where $p_{m,k}$ is the transmission power for UE $m$ in AP $k$}. It can be easily proved that a problem with the objective of maximizing EE (which is a ratio of total rate to the power consumption) is equivalent to a multi-objective optimization which the objectives are maximizing total rate (the nominator of the EE) and minimizing total power consumption (the denominator of the EE) \cite{Zarandi2020}. Therefore, invoking this property, we formulate a multi-objective optimization problem in the next section.

    \vspace{-5mm}
	\subsection{Multi-objective Problem Formulation}
	In this section, we formulate an optimization framework  where the objective is to jointly maximize the SR and minimize the total transmitted data power, i.e. SP. The problem formulation is described as follows
	\begin{subequations}
	\begin{flalign} \label{P1obj}
		\textbf{P1 }&
		\begin{cases}
			\displaystyle
			\max_{\boldsymbol{\rho},\textbf{p},\boldsymbol{\delta}} \enskip \text{SR} = \displaystyle \sum\limits_{n=1}^{N}\sum\limits_{k=1}^{K}\sum\limits_{m=1}^{U_k} \rho_{m,k}^n R_{m,k}^n \\
			\displaystyle
			\min_{\boldsymbol{\rho},\textbf{p},\boldsymbol{\delta}} \enskip
			\text{SP} = \sum\limits_{k=1}^{K}\sum\limits_{m=1}^{U_k}p_{m,k}
		\end{cases}
		\\ 
		\text{\textbf{s.t.}} \quad 
		& \text{C1: } 
		\textcolor{black}{\sum\limits_{n=1}^{N} \rho_{m,k}^{n} R_{m,k}^{n}\geq R_{m}^{\text{QoS}}, \smallskip \forall k\in\mathscr{K}, \forall m \in \mathscr{U}_k,}  \\ 
		& \textcolor{black}{\text{C2: } 0 \leq \delta_{n,k}^{l} \leq 1, \quad \forall n \in \mathscr{N}, n \neq 1, \forall k \in \mathscr{K},} \label{delta_1} \\
		& \textcolor{black}{\text{C3: } 0 \leq \delta_{n,k}^{r} \leq 1, \quad \forall n \in \mathscr{N}, n \neq N, \forall k \in \mathscr{K},} \label{delta_2} \\
		& \text{C4: } \sum\limits_{m=1}^{U_k}\rho_{m,k}^n \leq  L_{n,k}, \quad \forall n \in  \mathscr{N},\forall k \in \mathscr{K}, \\
		&  \text{C5: }
		\sum\limits_{n=1}^{N}\rho_{m,k}^n \leq 1, \quad \forall m \in \mathscr{U}_k, \forall k \in \mathscr{K}, \\ 
		%    & \text{C5: }
		%    \sum\limits_{n=1}^{N}\text{B}_n - \sum\limits_{n=1}^{N-1}\delta_{n,k}\text{B}_{n+1} = W, \quad \forall k \in \mathscr{K}\\
		& \text{C6: } p_{m,k} \leq P_{m}^{\text{max}},  \quad \forall m \in \mathscr{U}_k, \forall k \in \mathscr{K}, \label{Power_UL} \\ 
		& \text{C7: }
		\rho_{m,k}^n \in \{0,1\}, \quad \forall m \in \mathscr{U}_k, \forall k \in \mathscr{K}, \forall n \in \mathscr{N}.
	\end{flalign} 
	\end{subequations}
	where C1 denotes the QoS constraint guaranteeing minimum rate of each user, C2 and C3 are the amount of allowed inter-cluster partial ICI overlapping percentage. The constraint C4 shows that the total number of users in subband $n$ at AP $k$ is $L_{n,k}$ and C5 states that each UE must be allocated to only one subband, C6 is the maximum transmission power constraint for UE $m$. Problem \textbf{P1} is a type of mixed \textcolor{black}{integer} nonlinear and non-convex optimization which is intractable to solve. In order to find the Pareto-optimal solutions for \textbf{P1}, we inspire the Tchebycheff approach \cite{Ata_TWC,Marler2004} \textcolor{black}{which} is investigated in \textcolor{black}{the} next section comprehensively.
	\vspace{-3mm}
	\section{Problem Transformation and Solution}
	\subsection{Problem Transformation}
	\textcolor{black}{Inspired from weighted max-min formulation for multi-objective optimizations, to convert \textbf{P1} into a single objective optimization problem we employ Tchebycheff method as it provides complete Pareto-optimal solutions \cite{Marler2004}. Henceforth, by applying this method, the  multi-objective optimization problem \textbf{P1} is be transformed as follows:}
	\begin{flalign} \label{P2obj}
		\textbf{P2 }&
		\min_{\boldsymbol{\rho},\textbf{p},\boldsymbol{\delta},\lambda} \enskip \lambda
		\\
		\text{\textbf{s.t.}} \quad \nonumber
		&\textcolor{black}{\tilde{\text{C}}\text{1: }\omega(U_1^* - U_1(\boldsymbol{\rho},\textbf{p},\boldsymbol{\delta})) \leq \lambda,} \\\nonumber
		&\textcolor{black}{\tilde{\text{C}}\text{2: }(1-\omega)\left(U_2(\boldsymbol{\rho},\textbf{p},\boldsymbol{\delta}) - U_2^*\right) \leq \lambda, }\text{ C1--C7},
	\end{flalign} 
	\textcolor{black}{where $\lambda$ is an auxiliary parameter, $\omega$ is the non-negative weight typically set by a decision maker, and $U_1(\boldsymbol{\rho},\textbf{p},\boldsymbol{\delta}) = \text{SR}$ where SR is the sum-rate,  $U_2(\boldsymbol{\rho},\textbf{p},\boldsymbol{\delta}) =  \text{SP}$ in which SP is the sum-power defined earlier.}
% 	\begin{flalign}
% 		\label{U_i} 
% 		U_i(\boldsymbol{\rho},\textbf{p},\boldsymbol{\delta}) = 
% 		\begin{cases}
% 			\text{SR}, &\smallskip i=1, \\    
% 			\sum\limits_{k=1}^{K}\sum\limits_{m=1}^{U_k}p_{m,k}, &\smallskip i=2.\\
% 		\end{cases}
% 	\end{flalign}
	Furthermore, $U_i^*$ is the utopia point \cite{Marler2004},	\textcolor{black}{for $i\in \{1, 2\}$ obtained by solving the single-objective problems.}

	The problem \textbf{P2} is still non-convex due to binary \textcolor{black}{constraint C7 and non-convex constraint C1. To address this issue, first, we relax the binary constraint C7 then, add a regulation term;} next, we combine $p_{m,k}$ and $\rho_{m,k}^n$ and introduce new \textcolor{black}{constraint} $\tilde{p}_{m,k}^n \leq \rho_{m,k}^nP^{\text{max}}_m$ to be replaced with \textcolor{black}{C6}, where it is equal to zero when $\rho_{m,k}^n=0$. Therefore, the problem formulation will be
	\setlength{\abovedisplayskip}{3pt}
    \setlength{\belowdisplayskip}{\abovedisplayskip}
    \setlength{\abovedisplayshortskip}{0pt}
    \setlength{\belowdisplayshortskip}{3pt}
	\begin{flalign} \label{P3obj}
		\textbf{P3 }&
		\max_{\boldsymbol{\rho},\tilde{\textbf{p}},\boldsymbol{\delta},\lambda} \enskip -\lambda + \underset{\textcolor{black}{\text{Regulation term}}}{\underbrace{\sum\limits_{n=1}^{N}\sum\limits_{k=1}^{K}\sum\limits_{m=1}^{U_k} \alpha\big(({\rho_{m,k}^n})^2-\rho_{m,k}^n\big) }}   \\
		\text{\textbf{s.t.}} \quad 
		& \tilde{\text{C}}\text{1, } \tilde{\text{C}}\text{2, C1--C5}, \nonumber \\
		& \tilde{\text{C}}\text{6: }
		\tilde{p}_{m,k}^n \leq \rho_{m,k}^nP^{\text{max}}_m, \quad \forall k \in \mathscr{K}, \forall m \in \mathscr{U}_k, \forall n \in \mathscr{N}, \nonumber\\
		& \tilde{\text{C}}\text{7: }
		\rho_{m,k}^n \in [0,1], \quad \forall k \in \mathscr{K}, \forall m \in \mathscr{U}_k, \forall n \in \mathscr{N}, \nonumber
	\end{flalign} 
	where a regulation term is added to the objective function with parameter $\alpha\gg1$ that forces the relaxed variables $\rho_{m,k}^n$ to be approximately close to either \textcolor{black}{zero or one. \textcolor{black}{For a sufficiently large value of $\alpha$, optimization problem \textbf{P3} is equivalent to \textbf{P2} as both problems attain the same optimal values \cite{Ata_TWC}.} In other words, the parameter $\alpha$ controls the importance of the regulation term penalty in the objective function, however its value during the simulations shall be chosen properly, since for an exceedingly large $\alpha$, the regulation term would dominate the objective function. Note that $\tilde{\textbf{p}}$ is the vector representation of  the new variables} $\tilde{p}_{m,k}^n$ for $\forall (m,k,n)\in\mathscr{U}_k\times\mathscr{K}\times\mathscr{N}$.
    
    \vspace{-3mm}
	\subsection{Solution Approach (Monotonic Optimization)}
	In this section, our aim is to convert \textcolor{black}{problem \textbf{P3}} into the caconical form of the monotonic optimization framework \cite{Tuy2000,Tuy2005,Zappone2017}. It is observed that the objective and the constraints in \textbf{P3} are not strictly increasing, however they can be written in terms of difference of increasing functions (DIF). \textcolor{black}{Thus, the following theorem is expressed for optimization problems incorporating DIFs.
	\begin{theorem}
		The following optimization problem 
		\begin{flalign}
			\boldsymbol{\mathscr{P}}\text{: }  \max_{\textbf{x}} \enskip f(\textbf{x})-g(\textbf{x})  \nonumber \\ 
			\textbf{\text{s.t.}} \quad \textbf{x} \in \boldsymbol{\Xi}\cap\boldsymbol{\Xi}_c, \nonumber
		\end{flalign}
		where  $\boldsymbol{\Xi}$,  $\boldsymbol{\Xi}_c$ denote the normal and co-normal sets, respectively and $f(.)$ and $g(.)$ are both increasing functions in $[\textbf{0},\textbf{b}]$  is a class of monotonic optimization problem.
		\label{theo1}
	\end{theorem}
	\begin{proof}
	Please refer to \cite{Zappone2017}.
	\end{proof}}
	
	\textcolor{black}{Because the the objective and the constraints in \textbf{P3} do not explicitly indicate monotonicity, our aim is to reformulate \textbf{P3} to explore some hidden monotonicity.} To do so, we first define the following functions
	\begin{flalign}
		\nonumber
		& q_0^+(\boldsymbol{\rho}) =  \sum\limits_{n=1}^{N}\sum\limits_{k=1}^{K}\sum\limits_{m=1}^{U_k} \alpha{(\rho_{m,k}^n})^2,  
		q_0^-(\boldsymbol{\rho},\lambda) = \sum\limits_{n=1}^{N}\sum\limits_{k=1}^{K}\sum\limits_{m=1}^{U_k} \alpha{\rho_{m,k}^n} + \lambda, \\
	    \nonumber
		& q^+_1(\tilde{\textbf{p}},\boldsymbol{\delta},\lambda)  = \omega   \sum\limits_{n=1}^{N}\sum\limits_{k=1}^{K}\sum\limits_{m=1}^{U_k} \text{B}_n\log_2\Big(\displaystyle \tilde{p}_{m,k}^n|g_{m,k}^{k,n}|^2 + \text{I}_\text{\textcolor{black}{Total}}^{m,n,k}\Big) + \lambda,\\ \nonumber
		& q^-_1(\tilde{\textbf{p}},\boldsymbol{\delta})=  \omega   \sum\limits_{n=1}^{N}\sum\limits_{k=1}^{K}\sum\limits_{m=1}^{U_k} \text{B}_n\log_2\Big(\displaystyle  \text{I}_\text{\textcolor{black}{Total}}^{m,n,k}\Big) + \omega U_1^*,\\
		& q^+_2(\tilde{\textbf{p}},\boldsymbol{\delta})=   \min_{m,n,k} \Bigg\{ \text{B}_n\log_2\Big(\displaystyle \tilde{p}_{m,k}^n|g_{m,k}^{k,n}|^2+\displaystyle  \text{I}_\text{\textcolor{black}{Total}}^{m,n,k}\Big)  +\sum\limits_{\substack{n'=1,\\n'\neq n}}^{N}\sum\limits_{\substack{k'=1,\\k'\neq k}}^{K}\sum\limits_{\substack{m'=1,\\m'\neq m}}^{U_{k'}} \bigg(\text{B}_{n'}\log_2\Big(\displaystyle \displaystyle  \text{I}_\text{\textcolor{black}{Total}}^{m',n',k'}\Big) + R_{m'}^{\text{QoS}} \bigg)\Bigg\},  \nonumber\\ \nonumber
		& q_2^-(\tilde{\textbf{\text{p}}},\boldsymbol{\delta}) =  \sum\limits_{{n=1}}^{N}\sum\limits_{{k=1}}^{K}\sum\limits_{{m=1}}^{U_k} \bigg(\text{B}_{n}\log_2\Big(\displaystyle \displaystyle  \text{I}_\text{\textcolor{black}{Total}}^{m,n,k}\Big) + R_{m}^{\text{QoS}} \bigg), \\
		\nonumber
		& q_3^+(\lambda) =  \lambda + (1-\omega)U_2^*, \enskip  q_3^-(\tilde{\textbf{\text{p}}}) =  (1-\omega)\sum\limits_{k=1}^{K}\sum\limits_{m=1}^{U_k}\tilde{p}_{m,k}^n, \\ \nonumber
		& q_4^+(\tilde{\textbf{\text{p}}},\boldsymbol{\rho})=  \min_{m,n,k} \bigg\{\rho_{m,k}^nP^{\text{max}}_m+\sum\limits_{\substack{k'=1,\\k'\neq k}}^{K}\sum\limits_{\substack{m=1,\\m'\neq n}}^{U_{k'}}\tilde{p}_{m',k'}^{n'}\bigg\}, \\
		& q_4^-(\tilde{\textbf{\text{p}}})=  \sum\limits_{k=1}^{K}\sum\limits_{m=1}^{U_k}\tilde{p}_{m,k}^n,\nonumber
	\end{flalign}
	where we observe that $q_0^{\pm}(.), q_1^{\pm}(.), q_2^{\pm}(.), q_3^{\pm}(.), q_4^{\pm}(.)$ are increasing functions. \textcolor{black}{Now, the constraints can be written in the form of difference of increasing functions.} To do so, we define $\tilde{\textbf{p}}_{\text{max}}$ as the tensor containing maximum transmit powers (i.e. $P_m^{\text{max}}$), $ \boldsymbol{\rho}_{\text{max}}=\textbf{1}$ for $\forall k \in \mathscr{K}, m \in \mathscr{U}_k$ and $\forall n \in \mathscr{N}$ and $\boldsymbol{\delta}_{\text{max}}=\textbf{1}, \forall k \in \mathscr{K}, n \in \mathscr{N}$; moreover we set $\lambda_{\text{max}} = \Lambda$. \textcolor{black}{It can be inferred that the problem \textbf{P3} is an optimization problem where DIFs are used in the objective as well as constraints. Thus, we can  convert the problem into an MO by employing Theorem \ref{theo1}. } In order to proceed, we define the auxiliary variables	$t=q_0^-(\boldsymbol{\rho}_{\text{max}},\Lambda)-q_0^-(\boldsymbol{\rho},\lambda)$, $w =  q^-_1(\tilde{\textbf{p}}_{\text{max}},\boldsymbol{\delta}_{\text{max}}) - q^-_1(\tilde{\textbf{p}},\boldsymbol{\delta})$, $l=q^-_2(\tilde{\textbf{p}}_{\text{max}}$,  $\boldsymbol{\delta}_{\text{max}}) - q^-_2(\tilde{\textbf{p}},\boldsymbol{\delta})$, $u=q_3^-(\tilde{\textbf{\text{p}}}_{\text{max}}) -q_3^-(\tilde{\textbf{\text{p}}})$, $v=q_4^-(\tilde{\textbf{\text{p}}}_{\text{max}}) -q_4^-(\tilde{\textbf{\text{p}}})$.
% 	$t,w,l,u,v$ as
%     \textcolor{black}{\begin{subequations}
%     \nonumber 
%     	\begin{flalign}
%     		t=&q_0^-(\boldsymbol{\rho}_{\text{max}},\Lambda)-q_0^-(\boldsymbol{\rho},\lambda),  \\
%     		w = & q^-_1(\tilde{\textbf{p}}_{\text{max}},\boldsymbol{\delta}_{\text{max}}) - q^-_1(\tilde{\textbf{p}},\boldsymbol{\delta}),  \\
%     		l=&q^-_2(\tilde{\textbf{p}}_{\text{max}},\boldsymbol{\delta}_{\text{max}}) - q^-_2(\tilde{\textbf{p}},\boldsymbol{\delta}),  \\
%     		u=&q_3^-(\tilde{\textbf{\text{p}}}_{\text{max}}) -q_3^-(\tilde{\textbf{\text{p}}}), \enskip 
%     		v=q_4^-(\tilde{\textbf{\text{p}}}_{\text{max}}) -q_4^-(\tilde{\textbf{\text{p}}}),
%     	\end{flalign}
%     \end{subequations}}
	Then, \textcolor{black}{the problem \textbf{P3}} will be transformed into 
	\setlength{\abovedisplayskip}{3pt}
    \setlength{\belowdisplayskip}{\abovedisplayskip}
    \setlength{\abovedisplayshortskip}{0pt}
    \setlength{\belowdisplayshortskip}{3pt}
	\begin{subequations}
		\begin{flalign}
			\textbf{P4 }&
			\max_{\substack{\boldsymbol{\rho},\tilde{\textbf{p}},\boldsymbol{\delta},\lambda,\\ t, w, l, u, v}} \enskip q_0^+(\boldsymbol{\rho}) +t    \nonumber \\
			\text{\textbf{s.t.}} \quad 
			& \text{ N1: } 0 \leq t+q_0^-(\boldsymbol{\rho},\lambda) \leq q_0^-(\boldsymbol{\rho}_{\text{max}},\Lambda), \nonumber \\
			&\text{ N2: } 0 \leq t \leq q_0^-(\boldsymbol{\rho}_{\text{max}},\Lambda)-q_0^-(\textbf{0},0), \nonumber\\
			\tilde{\text{C}}\text{1: } &
			\begin{cases}
				\text{Co-N1: }q^+_1(\tilde{\textbf{p}},\boldsymbol{\delta},\lambda)+w \geq q^-_1(\tilde{\textbf{p}}_{\text{max}},\boldsymbol{\delta}_{\text{max}}), \nonumber\\ 
				\text{N3: }0 \leq w \leq q^-_1(\tilde{\textbf{p}}_{\text{max}},\boldsymbol{\delta}_{\text{max}}) -q^-_1(\textbf{0},\textbf{0}),  \nonumber\\
				\text{N4: }0 \leq w+q^-_1(\tilde{\textbf{p}},\boldsymbol{\delta}) \leq q^-_1(\tilde{\textbf{p}}_{\text{max}},\boldsymbol{\delta}_{\text{max}}), 
			\end{cases}
			\nonumber \\
			\text{C1: }&
			\begin{cases}
				\text{Co-N2: }q^+_2(\tilde{\textbf{p}},\boldsymbol{\delta}) +l \geq q^-_2(\tilde{\textbf{p}}_{\text{max}},\boldsymbol{\delta}_{\text{max}}), \nonumber\\
				\text{N5: }0 \leq l \leq q^-_2(\tilde{\textbf{p}}_{\text{max}},\boldsymbol{\delta}_{\text{max}})-q^-_2(\textbf{0},\textbf{0}), \nonumber\\
				\text{N6: }0 \leq l+q^-_2(\tilde{\textbf{p}},\boldsymbol{\delta}) \leq q^-_2(\tilde{\textbf{p}}_{\text{max}},\boldsymbol{\delta}_{\text{max}}),
			\end{cases}
			\nonumber \\
			\tilde{\text{C}}\text{2: }&
			\begin{cases}
				\text{Co-N3: }q_3^+(\lambda)+u\geq q_3^-(\tilde{\textbf{\text{p}}}_{\text{max}}),\nonumber \\
				\text{N7: }0 \leq u \leq q_3^-(\tilde{\textbf{\text{p}}}_{\text{max}})-q_3^-(\textbf{0}), \nonumber\\
				\text{N8: }0 \leq u+q_3^-(\tilde{\textbf{\text{p}}}) \leq q_3^-(\tilde{\textbf{\text{p}}}_{\text{max}}),
			\end{cases}
			\nonumber\\
			\tilde{\text{C}}\text{6: }&
			\begin{cases}
				\text{Co-N4: }q_4^+(\tilde{\textbf{\text{p}}},\boldsymbol{\rho})+v\geq q_4^-(\tilde{\textbf{\text{p}}}_{\text{max}}),\nonumber\\
				\text{N9: }0 \leq v \leq q_4^-(\tilde{\textbf{\text{p}}}_{\text{max}}) - q_4^-(\textbf{0}),\nonumber\\
				\text{N10: }0 \leq v + q_4^-(\tilde{\textbf{\text{p}}}) \leq q_4^-(\tilde{\textbf{\text{p}}}_{\text{max}}),
			\end{cases}
			\nonumber\\
			%    &\text{Co-N5: }\sum\limits_{m=1}^{U_k}\rho_{m,k}^n \geq  L_{n,k}, \quad \forall n \in  \mathscr{N},\forall k \in \mathscr{K}, \\
			%		\text{C4: } &
			%		\text{   N11: } \sum\limits_{m=1}^{U_k}\rho_{m,k}^n \leq  L_{n,k}, \quad \forall n \in  \mathscr{N},\forall k \in \mathscr{K}, 
			%		\nonumber\\
% 			{\text{C}}\text{5: }&
% 			\begin{cases}
% 				\text{Co-N5: }
% 				\sum\limits_{n=1}^{N}\rho_{m,k}^n \geq 1, \quad \forall m \in \mathscr{U}_k, \forall k \in \mathscr{K},\nonumber \\ 
% 				\text{N12: }
% 				\sum\limits_{n=1}^{N}\rho_{m,k}^n \leq 1, \quad \forall m \in \mathscr{U}_k, \forall k \in \mathscr{K},
% 			\end{cases}
% 			\nonumber \\
			&\text{C2--C5, }\tilde{\text{C}}\text{7,} \nonumber
			%    & \text{Co-N7: }
			%    \sum\limits_{n=1}^{N}\text{B}_n - \sum\limits_{n=1}^{N-1}\delta_{n,k}\text{B}_{n+1} \leq W, \quad \forall k \in \mathscr{K}\\
			%    & \text{N13: }
			%    \sum\limits_{n=1}^{N}\text{B}_n - \sum\limits_{n=1}^{N-1}\delta_{n,k}\text{B}_{n+1} \geq W, \quad \forall k \in \mathscr{K}.
		\end{flalign}    
	\end{subequations}
 	\textcolor{black}{where N1--N10 and C2--C5 and $\tilde{\text{C}}$7 are the constraints that build normal set denoted by $\boldsymbol{\Xi}$ and Co-N1--Co-N4 indicate the constraints constructing a co-normal set shown as $\boldsymbol{\Xi}_c$. Therefore, \textbf{P4} is a monotonic optimization in standard canonical form \cite{Tuy2000}. Henceforth, the optimal solution of \textbf{P4} lies on the boundary of the feasible set which is defined as $\boldsymbol{\Xi}\cap\boldsymbol{\Xi}_c$. A well-known method to solve MO problems is polyblock algorithm \cite{Tuy2000,Tuy2005,Zappone2017} in which searches the upper boundary vertex set (edges of the feasible region). The main advantage of employing polyblock algorithm is that the feasible search set is reduced to look up in the boundaries, whereas in the other methods, usually the whole feasible region is exhaustively searched which is intractable and not practical in most cases.} 
 	
 	\textcolor{black}{The computational complexity of the polyblock algorithm is severely based on the form of the functions providing number of variables and the normal, co-normal sets.~Note that, the size of the vertex set can grow exponentially over iterations, however, some of the vertices are not needed in the computation, and therefore can be safely discarded that accelerate the algorithm speed.
%  	This not only leads to a high computational complexity to find the optimal vertex, but also may cause memory overflow problems.~Also, an important issue of such an algorithm is its computational complexity, which heavily relies on the complexity of computing the upper boundary projection vertex.
 	Assuming the offered algorithm includes following steps starting from a hyper-rectangle that encloses $\boldsymbol{\Xi}\cap\boldsymbol{\Xi}_c$: In the first step, we obtain the best vertex which its projection belongs to the normal set. Next, we obtain the projection of selected vertex on the normal set by using bisection algorithm. The new vertex set is found based on the projection of the vertex on the normal set upperboundary. Then, the improper vertexes that do not satisfy the co-normal constraints are removed. The algorithm continues until a convergence criteria and the best candidate vertex point is reported. While assuming that the dimensions of the optimization problem is $M_1$, the number of iterations in overall polyblock algorithm to converge is $M_2$ and the number of iteration in bisection algorithm for the projection of each vertex is $M_3$.  Then, the complexity order is expressed as $\mathcal{O}\left(M_2(M_2\times M_1+M_3)\right)$.}
	
	\vspace{-3mm}
	\section{Numerical Results}
	In what follows, we evaluate the proposed optimizations by using the transformed monotonic problem, \textbf{P4}. Table \ref{table2} shows the considered chosen values for the parameters of the network.  \textcolor{black}{To avoid numerical issues during simulations and for the sake of simplicity, we assume that $\delta_{n,k}^r=\delta_{n+1,k}^l$, meaning that the percentage of overlapping between left-hand-side of channel $n+1$ is the same as right-hand-side overlapping at channel $n$.  In this way, linear expressions w.r.t. $\delta_{n,k}^r$ and $\delta_{n,k}^l$ will be achieved in the interference terms defined in \eqref{IntraInt} which yields $(\sqrt{\delta_{n,k}^r}+\sqrt{\delta_{n+1,k}^l})^2=4\delta_{n,k}^r$.}
	\begin{table}[t]
		\caption{Simulation Parameters.}
		\centering
		\textcolor{black}{\begin{tabular}{ l l}
			\hline
			Parameter & Default value   \\ \hline
			AP coverage diameter & 200 m \\
			Number of APs ($K$) & 2   \\
			Number of UEs in each AP & 6 \\
			Number of subbands ($N$) & 4 \\
			Subband bandwidth & 180 kHz \\ 
			Number of UEs in a subband & 2 \\
			Maximum data transmit power $P^{\text{max}}_m$, $\forall m$  & 200 mW  \\ 
			Circuit power consumption  & 30 mW \\ 
			Noise spectral density & $-174$ dBm/Hz \\	
			Minimum data rate requirement ${R}^{\text{QoS}}_m $, $\forall m$   & 0.1 bps/Hz \\ 
			Path loss model ($d$: distance)  & $34.53+38\log_{10}(d)$ [dB] \\
			Receiver noise figure (NF)  & 3  dB  \\ \hline
		\end{tabular}}
		\label{table2}
	\end{table}
	Poisson point process (PPP) is leveraged to generate the users' location in 400$\times$400 m$^2$ area where \textcolor{black}{one AP is located at $(100,100)$ m and the second AP is placed at $(300, 100)$ m}. To have benchmarks for comparisons, we consider three scenarios, optimizing $\delta_{n,k}^r$ and $\delta_{n,k}^l$ for $\forall n \in \mathscr{N}, \forall k \in \mathscr{K}$ where we name it as \textcolor{black}{POD}, the next scenario is optimizing
	$\delta_{n,k}^r=\delta$ and $\delta_{n,k}^l=\delta$ for $\forall n \in \mathscr{N}, \forall k \in \mathscr{K}$ where it is named as \textcolor{black}{NPOD} and the scenario where $\delta_{n,k}^r=\delta_{n,k}^l=0$ which is NOMA-OFDM. \textcolor{black}{The SE$=\frac{\text{SR}}{W}$ curves for different multiple access methods are shown in Fig. \ref{fig:3}. Besides, the sum power values are shown for different schemes. We can see that the power consumption for different multiple access methods is roughly the same except for OFDMA. Furthermore, the importance of the sum power objective reduces in case of increasing $\omega$, and hence, the power consumption increases as well. Moreover, when we optimize the percentage of overlapping among channels individually in POD scheme, higher amount of rate is achieved compared with the traditional NOMA-OFDM case due to the better utilization of the spectrum. Furthermore, we observe that the SE is an increasing function with respect to $\omega$ as the resource allocation emphasizes more on SE maximization.}
	\begin{figure}[t]
		\centering
		\includegraphics[trim = 8.75cm 8.75cm 8.75cm 8.75cm,scale=0.6]{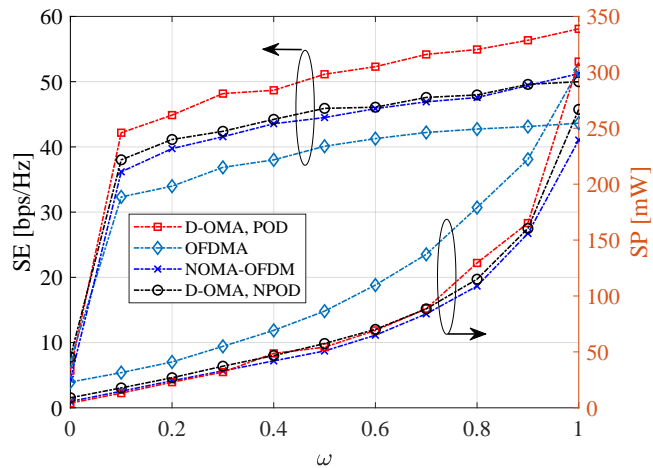}
		\caption{\textcolor{black}{Total SE and SP for different multiple access schemes.}}
		\label{fig:3}
	\end{figure}
	
	\textcolor{black}{It is observed that the D-OMA POD scheme outperforms other multiple access methods in Fig. \ref{fig:5}, where the EE is plotted as a function of SE. Also, it is inferred that the D-OMA POD scheme has higher SE for a fixed EE value, as a result it outperforms other scenarios. The major contribution of using the D-OMA POD scheme is that a higher amount of rate can be achieved without expanding the current available bandwidth. Since the D-OMA POD scheme achieves higher SE without increasing bandwidth, therefore, it would be interesting for Telecommunication operators. Because an operator leverages a small amount of bandwidth in the frequency spectrum to work with and service their users.}
	
	\begin{figure}[t]
		\centering
		\includegraphics[trim = 8.75cm 8.75cm 8.75cm 8.75cm,scale=0.6]{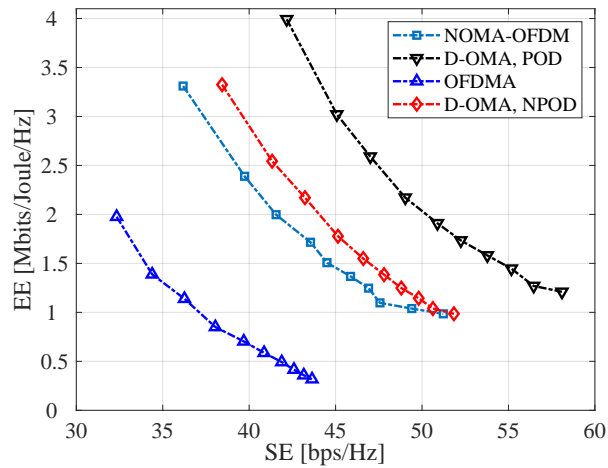}
		\caption{Total EE versus SE.}
		\label{fig:5}
	\end{figure}

	\textcolor{black}{Finally, Fig. \ref{fig:6} illustrates the performance of the D-OMA versus  NOMA-OFDM in terms of the number of UEs in the network. The weight factor is assumed to be $\omega=0.4$ and the number of UEs in each subband is assumed to be $L_{n,k}=10, \forall n, k$. It can be observed that in both multiple access methods, the network total spectral efficiency increases until it is saturated due to the increasing interference signals generated from having more UEs. Furthermore, it is inferred that the D-OMA POD method outperforms NOMA-OFDM in terms of total SE. This is because the overlapping ratios are optimized in D-OMA POD method with fixed available bandwidth so as to maximize the total SE of the network which highlights the significance of D-OMA.}
	\begin{figure}[t]
		\centering
		\includegraphics[trim = 8.75cm 8.75cm 8.75cm 8.75cm,scale=0.6]{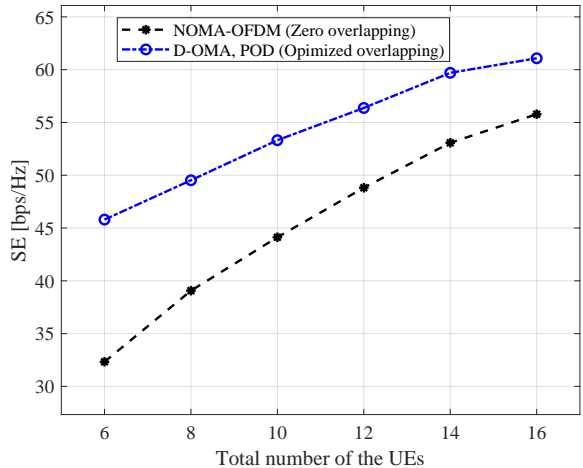}
		\caption{\textcolor{black}{Total network SE versus total number of the UEs.}}
		\label{fig:6}
	\end{figure}
	
	\vspace{-5mm}
	\section{Conclusion}
	In this study, we examined the performance of a D-OMA enabled network in the UL by proposing a multi-objective optimization framework. The the achievable rate expression in the uplink of a multi-AP network is identified in terms of the partial ICI which is controlled by the overlapping ratios between subbands. The network sum-rate and transmit power consumption of the UEs are considered as the objective functions where the variables are the subband associations and transmit power of the UEs as well as the overlapping ratios. To solve the proposed optimization problem, the binary decision variables are relaxed and regulation term is added to the objective function. Next, a method named as Tchebycheff is leveraged to transform the multi-objective optimization to a single-objective problem with the same constraints. Then, the problem is reformulated into a monotonic optimization framework by exploring the hidden monotonicity of the objective and constraints. The numerical results show that the novel D-OMA method outperforms other traditional multiple access methods such as OFDMA, NOMA-OFDM, and OMA. \textcolor{black}{As a future topic, the D-OMA method can be also elaborated for imperfect SIC case which is a practical scenario.}
	\vspace{-4.5mm}
	%	\newpage
	\bibliographystyle{IEEEtran}% such as plain
	\bibliography{main.bib} %such as MyReferences
	
\end{document}